\documentclass[nopreprintline,3p,twocolumn]{elsarticle}

\usepackage{lineno,hyperref}
\modulolinenumbers[5]
\usepackage{amsmath,amssymb}
\usepackage{enumitem}

\newtheorem{theorem}{Theorem}
\newtheorem{proposition}[theorem]{Proposition}
\newdefinition{definition}{Definition}
\newdefinition{example}{Example}
\newproof{proof}{Proof}
\newproof{proofsketch}{Proof sketch}

\newcommand{\Com}{\texttt{Com}}
\newcommand{\Skip}{\texttt{skip}}
\newcommand{\pop}{\texttt{pop}}
\newcommand{\push}{\texttt{push}}
\newcommand{\upds}{\texttt{upds}}
\newcommand{\frsp}{\texttt{frsp}}
\newcommand{\lat}{\texttt{lat}}
\newcommand{\snd}{\texttt{snd}}
\newcommand{\Top}{\texttt{top}}
\renewcommand{\L}{\mathtt{x}}
\newcommand{\R}{\mathtt{x}'}
\newcommand{\Ltop}{\mathtt{x}_{\Top}}
\newcommand{\tdollar}{\texttt{\$}}
\newcommand{\Tt}{\texttt{tt}}

\newcommand{\calP}{\mathcal{P}}
\newcommand{\calA}{\mathcal{A}}
\newcommand{\neXt}{\mathbin{\mathcal{X}}}
\newcommand{\Until}{\mathbin{\mathcal{U}}}
\newcommand{\N}{\mathbb{N}}

\newcommand{\ID}{\mathit{ID}}
\newcommand{\IDA}{\ID_{\!\calA}}
\newcommand{\Acc}{\mathit{Acc}}
\newcommand{\At}{\mathit{At}}
\newcommand{\com}{\mathit{com}}

\newcommand{\done}{\Rightarrow}
\newcommand{\COMP}{\circ}
\newcommand{\COMPT}{\circ_{\mathtt{T}}}
\newcommand{\COMPBL}{\odot}
\newcommand{\COMPBLT}{\odot_{\mathtt{T}}}
\newcommand{\EQj}[2]{({#1})^{=}_{#2}}

\def\putconfig(0,#1)#2#3#4{%
  \put(0,#1){\makebox(0,0){$#2$}}
  \put(30,#1){\makebox(0,0){$[#3]$}}
  \put(55,#1){\makebox(30,0)[r]{$#4$}}}
\def\putrule(0,#1)#2#3{%
  \put(100,#1){\makebox(0,18)[l]{$#2\,#3$}}}
\def\putruleP(0,#1)#2{%
  \putrule(0,#1){\Downarrow}{#2}}
\def\putruleA(0,#1)#2{%
  \putrule(0,#1){\!\rotatebox[origin=c]{-90}{$\vdash$}\,}{#2}}

\begin{document}

\begin{frontmatter}

\title{Reduction of
  Register Pushdown Systems with Freshness Property \\
  to Pushdown Systems
  in LTL Model Checking\tnoteref{t1}}

\author[1]{Yoshiaki Takata} 

\author[2]{Ryoma Senda}

\author[2]{Hiroyuki Seki}

\tnotetext[t1]{Funding:
  This work was partially supported by
  JSPS KAKENHI Grant Numbers JP19H04083 and JP20K20625.}

\address[1]{Kochi University of Technology,
  Tosayamada, Kami City, Kochi 782--8502, Japan}
\address[2]{Nagoya University,
  Furo-cho, Chikusa, Nagoya 464--8601, Japan}

\begin{abstract}
Pushdown systems (PDS) are known as an abstract model of
recursive programs, and
model checking methods for PDS have been studied.
Register PDS (RPDS) 
are PDS augmented by
registers to deal with data values from an infinite domain
in a restricted way.
A linear temporal logic (LTL) model checking method
for RPDS with regular valuations
has been proposed;
however,
the method requires the register automata (RA) used for
representing a regular valuation to be backward-deterministic.
This paper proposes another approach to the same problem,
in which
the model checking problem for RPDS is reduced
to that problem for PDS
by constructing a PDS
bisimulation equivalent
to a given RPDS\@.
The construction in the proposed method
is simpler than the previous model checking method
and does not require RAs deterministic or backward-deterministic,
and
the bisimulation equivalence
clearly guarantees the correctness of this reduction.
On the other hand,
the proposed method
requires every RPDS (and RA) to have
the freshness property,
in which whenever the RPDS updates a register with a data value
not stored in any register or the stack top,
the value should be fresh.
This paper also shows that this model checking
problem with regular valuations
defined by general RA is undecidable, and
thus the freshness constraint is essential in the proposed method.
\end{abstract}


\end{frontmatter}


\section{Introduction}

A pushdown system (PDS) is a pushdown automaton without input
and is well known as an abstract model of
recursive programs~\cite{ABE18,BEM97}.
The model checking problem for a PDS $\calP$, i.e.\
testing whether all runs of $\calP$ conforms to a given specification
$\varphi$,
has been studied for various logics such as
linear temporal logic (LTL) and branching-time temporal
logic~\cite{ABE18,BEM97,Wal96}.
In~\cite{EKS03},
the LTL model checking of PDS with
regular valuations was investigated and
shown to be decidable.
A valuation is a function that labels
each configuration in a run with
a subset of atomic propositions.
A regular valuation is a valuation that
labels each configuration $c$ of a PDS
with atomic propositions
depending on whether
the stack contents in $c$ match a given regular pattern.
%

Although PDS is a natural model of recursive programs,
it cannot deal with data values directly.
Register automata (RA) were introduced as an extension of
finite automata (FA) by adding the capability of
dealing with data values in a restricted way~\cite{KF94}.
RA has attracted attention as a formal model of navigational queries
concerning data values to structured data
such as XML documents~\cite{LV12}.
RA has also been considered as
a formal model of software systems with
unbounded resources
in e.g.\ runtime verification~\cite{GDPT13}
and reactive synthesis~\cite{EFR21}.
Similarly to extending FA to RA, PDS were extended to
pushdown register systems (PDRS)~\cite{MRT17} and
register pushdown systems (RPDS)~\cite{STS21},
both of them are equivalent each other,
and the reachability problem for them has been shown
to be EXPTIME-complete.

In previous work~\cite{STS21-IEICE-dec},
we have investigated the LTL model checking problem
with regular valuations for RPDS
and proposed
a method for solving the problem.
This method is a natural extension of the model checking method
for PDS in~\cite{EKS03}.
Similarly to
the method in~\cite{EKS03} that reduces the model checking problem
to the emptiness problem for B\"uchi pushdown systems,
in \cite{STS21-IEICE-dec} we introduced
B\"uchi register pushdown systems (BRPDS) and showed
a reduction of the model checking problem
to the emptiness problem for BRPDS\@.
Note that the regular valuations in~\cite{STS21-IEICE-dec} were
defined in terms of
backward-deterministic RA,
in which
every configuration
has a unique predecessor configuration for each
input data value.
This constraint is essential because
unlike FA, determinization is not possible for
general RA.

In this paper,
we work on the same problem
in a different approach.
We reduce the model checking problem for RPDS
to that problem for PDS
by constructing a PDS $\calP'$
bisimulation equivalent~\cite{CGKPV18-Bisim}
to a given RPDS $\calP$.
In the same way,
we also construct an FA $\calA'$ bisimulation equivalent
to an RA $\calA$ used for a regular valuation.
The bisimulation equivalences between
$\calP$ and $\calP'$ and between $\calA$ and $\calA'$
guarantee the correctness of this reduction.
%
%
The construction of a PDS bisimulation equivalent to an RPDS
is basically the same as
the one used in our recent work~\cite{STS21-ICTAC}
on reactive synthesis from specifications
given by deterministic register pushdown automata (DRPDA\@).
In~\cite{STS21-ICTAC}, we construct a pushdown automaton (PDA)
simulating a DRPDA for reducing
the realizability problem for DRPDA to
the same problem for PDA\@.
The proposed method in this paper is an application
of this construction to the model checking problem,
where not only RPDS but also RA used for a regular valuation
can be reduced to models without data values
in a uniform way,
and the correctness of the reduction of the model checking
problem can be proved easily based on bisimulation equivalence.

Another feature of the proposed method is that
it does not require RA used for regular valuations
to be deterministic or backward-deterministic.
However,
the proposed method requires
every RPDS (and RA) to have
the freshness property instead,
in which whenever the RPDS updates a register with a data value
not stored in any register or the stack top,
the value should be fresh, i.e.~not used before.
This paper also shows that the LTL model checking
problem for RPDS with regular valuations
defined by general RA is undecidable
(Theorem~\ref{th:undecidable}), and
this fact compels
the proposed method to require
the freshness property
instead of requiring RA
to be deterministic or backward-deterministic.

Advantages of the proposed method are summarized as follows:
(i) The construction in the proposed method
is much simpler 
than the method
in~\cite{STS21-IEICE-dec}.
The bisimulation equivalence between an RPDS and
a PDS clearly
shows the correctness of the reduction.
We have proved this bisimulation equivalence
using a proof assistant software Coq~\cite{BC04}.
(ii) By the proposed method,
we can apply existing model checking tools for PDS
to solving the model checking problem for RPDS\@.
(iii) The method does not require RA used for regular valuations
to be deterministic or backward-deterministic,
though it requires the freshness property instead.
Showing the undecidability of the model checking problem
in the general case
is another contribution of the paper.

The rest of the paper is organized as follows.
We define basic notions in Section~\ref{sec:preliminaries},
RPDS and RA in Section~\ref{sec:rpds}, and
the LTL model checking problem for RPDS with regular valuations
in Section~\ref{sec:modelcheck}.
In Section~\ref{sec:modelcheck},
we also show that
this problem is undecidable in general.
In Section~\ref{sec:freshness},
we introduce the freshness property and
redefine the semantics of RPDS so that
every RPDS has this property.
In Section~\ref{sec:rpds-pds},
we show the construction of a PDS bisimulation equivalent to
a given RPDS,
which is the main part of the proposed method.
We conclude the paper
in Section~\ref{sec:conclusion}.

\section{Preliminaries}
\label{sec:preliminaries}

Let $\N=\{1,2,\ldots\}$, $\N_0=\{0\}\cup\N$,
and $[n]=\{1,2,\ldots,n\}$ for $n\in\N$.
We assume a countable set $D$ of \emph{data values}.
For a given $k\in\N_0$,
a mapping $\theta:[k]\to D$ is called an \emph{assignment}
(of data values to $k$ registers).
Let $\Theta_k$ be the set of assignments to $k$ registers.
Sometimes we consider an assignment $\theta\in\Theta_k$
as the set of assigned data values;
e.g., $d\in\theta$ means $d=\theta(i)$ for some $i\in[k]$.

For a set $A$,
let $A^*$ and $A^{\omega}$ be the sets of
finite and infinite words over $A$, respectively.
Let $A^{\infty}=A^*\cup A^{\omega}$.
For a word
$\alpha\in A^{\infty}$,
let $\alpha(i)\in A$ be the $i$-th element of $\alpha$ 
and $\alpha(i\,{:})=\alpha(i)\alpha(i+1)\ldots$ for $i\ge0$.
Let $\snd$ be the function over pairs that gives the second element
of a pair;
i.e., $\snd((a,b))=b$. 
For a word $w=(a_0,b_0)(a_1,b_1)\ldots$ over pairs,
let $\snd(w)=b_0b_1\ldots$\,.

For a relation $\done$,
let $\done^*$ be the reflexive transitive closure of~$\done$.

\subsection{Linear Temporal Logic (LTL)}

The definition of LTL formulas we used is the same
as~\cite{EKS03}.
Let $\At$ be a finite set of \emph{atomic propositions},
and let $\Sigma=2^{\At}$.
An \emph{LTL formula} over $\At$ is given by the following syntax:
\begin{equation*}
  \varphi ::= \Tt \mid A \mid \neg\varphi \mid \varphi_1\land\varphi_2
  \mid \neXt\varphi \mid \varphi_1\Until\varphi_2
\end{equation*}
where $A\in\At$.
For an infinite word $w\in\Sigma^{\omega}$,
the satisfaction relation $\models$ is defined as follows:
\begin{alignat*}{2}
  w &\models\Tt, \\
  w &\models A &\iff\ & A\in w(0), \\
  w &\models\neg\varphi &\iff\ & w\not\models\varphi, \\
  w &\models\varphi_1\land\varphi_2
    &\iff\ & w\models\varphi_1 \text{ and } w\models\varphi_2, \\
  w &\models\neXt\varphi &\iff\ & w(1\,{:})\models\varphi, \\
  w &\models\varphi_1\Until\varphi_2 \
    &\iff\ &
    \exists j:(w(j\,{:})\models\varphi_2) \\
    &&& {}\land (\forall i<j:w(i\,{:})\models\varphi_1).
\end{alignat*}
We also define $\Diamond\varphi = \Tt\Until\varphi$ and
$\Box\varphi = \neg\Diamond(\neg\varphi)$.

\subsection{Pushdown systems and finite automata}

For a finite set $\Gamma$, we define the set
$\Com(\Gamma)$ of \emph{commands} over $\Gamma$ as
$\Com(\Gamma) = \{\pop,\Skip\}\cup
\{\push(\gamma)\mid \gamma\in\Gamma\}$.

\begin{definition}
  A \emph{pushdown system} (PDS)
  $\calP$ over a stack alphabet $\Gamma$
  is a pair $(P,\Delta)$, where
  $P$ is a finite set of \emph{states} and
  $\Delta\subseteq P\times\Gamma\times P\times\Com(\Gamma)$
  is a set of \emph{transition rules}.
  We write an element $(p,\gamma,q,\com)\in\Delta$ as
  $(p,\gamma)\to(q,\com)$ for readability.

  Let $\ID_{\calP}=P\times\Gamma^*$ and call each element of
  $\ID_{\calP}$
  an \emph{instantaneous description} (ID) of $\calP$.
  The transition relation $\done_{\calP}$ of $\calP$
  is the smallest relation over $\ID_{\calP}$
  satisfying the following inference rule,
  where
  $\upds'(\gamma v,\com') = v$, $\gamma v$, or $\gamma'\gamma v$
  if $\com'=\pop$, $\Skip$, or $\push(\gamma')$, respectively.
\begin{equation*}
\begin{array}{l}
(p,\gamma)\to(q,\com')\in\Delta
\\ \hline
(p,\gamma v)\done_{\calP'} (q,\upds'(\gamma v,\com'))
\end{array}
\end{equation*}

  A \emph{run} of $\calP$ is a sequence $\rho\in\ID_{\calP}^{\omega}$
  such that $\rho(i)\done_{\calP}\rho(i+1)$ for $i\ge0$.
\end{definition}

We define nondeterministic finite automata as follows,
which is used
for representing a (regular) subset of IDs of some PDS\@.

\begin{definition}
  A \emph{nondeterministic finite automaton} (NFA)
  $\calA$ over an alphabet $\Gamma$ is a quadruple
  $(Q,I,F,\delta)$,
  where $(Q,\delta)$ is a PDS over $\Gamma$ where
  $\delta$ consists of pop rules only,
  $I\subseteq Q$ is a set of \emph{initial states}, and
  $F\subseteq Q$ is a set of \emph{final states}.
  We call $(Q,\delta)$ the \emph{base PDS} of $\calA$.
  The set $\IDA$ of IDs and the transition relation
  of $\calA$ are the same as those of its base PDS\@.
  We write the transition relation of $\calA$ as $\vdash_{\calA}$.
  The \emph{language} $L(\calA)$ of $\calA$ is a subset of $\IDA$
  defined as
  $L(\calA)=\{(p,w)\in I\times\Gamma^*\mid (p,w)\vdash_{\calA}^*
  (q,\varepsilon)$ for some $q\in F\}$.
\end{definition}

When we use an NFA $\calA$ for representing a subset of IDs
of a PDS $\calP=(P,\Delta)$,
we let the set of the initial states of $\calA$ be~$P$.

\begin{definition}
For a PDS $\calP=(P,\Delta)$,
we call a subset $C\subseteq\ID_{\calP}$ \emph{regular}
if there exists an NFA $\calA=(Q,P,F,\delta)$ that satisfies
$C=L(\calA)$.
\end{definition}

\subsection{Equivalence relations over registers}

Let $\Phi_{k}$ be the set of \emph{equivalence relations} over
the set of $2k+1$ symbols
$X_k = \{ \L_1,\ldots,\L_k,\allowbreak\R_1,\ldots,\R_k,
\allowbreak\Ltop \}$.
We write $a\equiv_{\phi}b$ and $a\not\equiv_{\phi}b$ to mean
$(a,b)\in\phi$ and $(a,b)\notin\phi$, respectively,
for $a,b\in X_k$ and $\phi\in\Phi_k$.
Intuitively,
each $\phi\in\Phi_k$ represents the equality and inequality
among the data values in the registers and the stack top,
as well as the transfer of the values in the registers
between two assignments.
Two assignments $\theta,\theta'\in\Theta_k$
and a value $d$ at the stack top satisfy $\phi$,
denoted as $(\theta,d,\theta')\models\phi$,
iff for $i,j\in[k]$,
\begin{alignat*}{2}
 \L_i\equiv_{\phi}\L_j &\Leftrightarrow \theta(i)=\theta(j),
 & \L_i\equiv_{\phi}\Ltop &\Leftrightarrow \theta(i)=d, \\
 \L_i\equiv_{\phi}\R_j &\Leftrightarrow \theta(i)=\theta'(j),\quad
 & \R_j\equiv_{\phi}\Ltop &\Leftrightarrow \theta'(j)=d , \\
 \R_i\equiv_{\phi}\R_j &\Leftrightarrow \theta'(i)=\theta'(j).
\end{alignat*}
We will use elements of $\Phi_k$ to specify
transition rules of a register pushdown system with $k$ registers
($k$-RPDS), defined in the next section.

Let $\Phi'_k$ be the set of equivalence relations over
the $k$ symbols $\{ \L_1,\ldots,\L_k \}$.
An assignment $\theta\in\Theta_k$ satisfies $\phi'\in\Phi'_k$,
denoted as $\theta\models\phi'$,
iff for $i,j\in[k]$,
$\L_i\equiv_{\phi'}\L_j \iff \theta(i)=\theta(j)$.
Elements of $\Phi'_k$ will be used to specify
accepting conditions of a register automaton with $k$ registers
($k$-RA\@).

Let $\lat:\Phi_k\to\Phi'_k$ be the function defined as:
$\forall i,j\in[k]:
\L_i\equiv_{\lat(\phi)}\L_j$ iff $\R_i\equiv_{\phi}\R_j$.

\section{Register pushdown systems and register automata}
\label{sec:rpds}

\begin{definition}
  A \emph{register pushdown system with $k$ registers}
  ($k$-RPDS) $\calP$ is a pair $(P,\Delta)$, where
  $P$ is a finite set of \emph{states} and
  $\Delta\subseteq P\times\Phi_{k}\times P\times\Com([k])$
  is a set of \emph{transition rules}.
  We write an element $(p,\phi,q,\com)\in\Delta$ as
  $(p,\phi)\to(q,\com)$ for readability.

  Let $\ID_{\calP}=P\times\Theta_k\times D^*$ and
  call each element of $\ID_{\calP}$
  an \emph{instantaneous description} (ID) of $\calP$.
  The transition relation $\done_{\calP}$ of $\calP$
  is the smallest relation over $\ID_{\calP}$
  satisfying the following inference rule,
  where
  $\upds(du,\theta',\com) = u$, $du$, or $\theta'(j)du$
  if $\com=\pop$, $\Skip$, or $\push(j)$, respectively.
\begin{equation*}
\begin{array}{l}
  (p,\phi)\to(q,\com)\in\Delta \quad
  (\theta,d,\theta')\models\phi
  \\ \hline
  (p,\theta,du)\done_{\calP} (q,\theta',\upds(du,\theta',\com))
\end{array}
\end{equation*}

  A \emph{run} of $\calP$ is a sequence $\rho\in\ID_{\calP}^{\omega}$
  such that $\rho(i)\done_{\calP}\rho(i+1)$ for $i\ge0$.
\end{definition}

\begin{example}
  \label{example:RPDS}%
  Let us consider 2-RPDS $\calP=(
  \{p_0, p_1,\allowbreak p_2\}$, $\{r_1,r_2,r_3,r_4,r_5\})$
  where
  \begin{align*}
    r_1&=(p_0,\phi_0)\to(p_1,\push(1)), \\
    r_2&=(p_1,\phi_1)\to(p_1,\push(1)), \\
    r_3&=(p_1,\phi_1)\to(p_1,\pop), \\
    r_4&=(p_1,\phi_2)\to(p_1,\pop), \\
    r_5&=(p_1,\phi_3)\to(p_2,\push(2)),
  \end{align*}
  and $\phi_0,\ldots,\phi_3\in\Phi_2$ are defined by the
  following quotient sets:
  \begin{align*}
    X_2/{\phi_0} &= \{ \{\L_1\}, \{\L_2,\R_2,\Ltop\}, \{\R_1\} \}, \\
    X_2/{\phi_1} &= \{ \{\L_1,\Ltop\}, \{\L_2,\R_2\}, \{\R_1\} \}, \\
    X_2/{\phi_2} &= \{ \{\L_1\}, \{\L_2,\R_2\}, \{\R_1,\Ltop\} \}, \\
    X_2/{\phi_3} &= \{ \{\L_1,\R_1\}, \{\L_2,\Ltop\}, \{\R_2\} \}.
  \end{align*}
  In this example, we let $[d_1,d_2]$ for $d_1,d_2\in D$ denote
  the assignment $\theta\in\Theta_2$ such that
  $\theta(1)=d_1$ and $\theta(2)=d_2$.
  Let $d_0,d_1,\ldots\in D$ represent distinct data values.
  Figure~\ref{fig:RPDS} shows a transition sequence of $\calP$
  starting from an ID $(p_0,[d_1,d_0],d_0)$.
  For example,
  we can apply $r_1$ to this starting ID and
  obtain $(p_1,[d_2,d_0],d_2d_0)$,
  because $\phi_0$ requires that
  the value of the second register
  before the transition is the same as the stack top,
  the second register is not changed by the transition,
  and the value of the first register
  after the transition is not equal to the value of
  any register before the transition.
\end{example}
\begin{figure}
  \centering
  \begin{picture}(155,120)
    \put(0,115){\makebox(0,0){\footnotesize state}}
    \put(30,115){\makebox(0,0){\footnotesize registers}}
    \put(55,115){\makebox(30,0){\footnotesize stack}}

    \putconfig(0,100){p_0}{d_1,d_0}{d_0}
    \putconfig(0,80){p_1}{d_2,d_0}{d_2\,d_0}
    \putconfig(0,60){p_1}{d_3,d_0}{d_3\,d_2\,d_0}
    \putconfig(0,40){p_1}{d_4,d_0}{d_2\,d_0}
    \putconfig(0,20){p_1}{d_2,d_0}{d_0}
    \putconfig(0,0){p_0}{d_2,d_5}{d_5\,d_0}

    \putruleP(0,80){\phi_0,\push(1)}
    \putruleP(0,60){\phi_1,\push(1)}
    \putruleP(0,40){\phi_1,\pop}
    \putruleP(0,20){\phi_2,\pop}
    \putruleP(0,0){\phi_3,\push(2)}
  \end{picture}
  \caption{A transition sequence of $\calP$ from
  $(p_0,[d_1,d_0],d_0)$.}
  \label{fig:RPDS}
\end{figure}

\begin{definition}
  A \emph{register automaton with $k$ registers} ($k$-RA)
  $\calA$ is a quadruple $(Q,I,\xi,\delta)$, where
  $(Q,\delta)$ is a $k$-RPDS
    where $\delta$ consists of pop rules only,
  $I\subseteq Q$ is a set of \emph{initial states}, and
  $\xi\subseteq Q\times\Phi'_k$ is a set of \emph{accepting conditions}.
  We call $(Q,\delta)$ the \emph{base RPDS} of $\calA$.
  The set $\IDA$ of IDs and the transition relation
  of $\calA$ are the same as those of its base RPDS\@.
  We write the transition relation of $\calA$ as $\vdash_{\calA}$.
  Let
  $\Acc_{\calA}=\{(p,\theta,\varepsilon)\in\IDA \mid$
  $\theta\models\psi$ for some $(p,\psi)\in\xi\}$.
  The \emph{language} $L(\calA)$ of $\calA$ is a subset of $\IDA$
  defined as
  $L(\calA)=\{(p,\theta,w)\in\IDA\mid
    p\in I$ and $(p,\theta,w)\vdash_{\calA}^* (q,\theta',\varepsilon)$
    for some
    $(q,\theta',\varepsilon)\in\Acc_{\calA}\}$.
\end{definition}

We write a transition rule $(q_1,\phi)\to(q_2,\pop)$ of an RA as
$(q_1,\phi)\to q_2$ for readability.

\begin{definition}
For a $k$-RPDS $\calP=(P,\Delta)$,
we call a subset $C\subseteq\ID_{\calP}$ \emph{regular}
if there exists a $k$-RA $\calA=(Q,P,\xi,\delta)$ that satisfies
$C=L(\calA)$.
\end{definition}

\begin{example}
  Let us consider 2-RA $\calA=(
  \{p_1, q_1, \allowbreak q_2\}$,
  $\{p_1\}$, $\{(q_2,\psi)\}$,
  $\{r_6, \allowbreak r_7, r_8\})$
  where
  \begin{alignat*}{2}
    r_6&=(p_1,\phi_1)\to q_1, \quad\ &
    r_8&=(q_1,\phi_3)\to q_2, \\
    r_7&=(q_1,\phi_4)\to q_1,
  \end{alignat*}
  $\phi_1$ and $\phi_3$ are the same as Example~\ref{example:RPDS},
  $\phi_4\in\Phi_2$ is defined by the quotient set
  \begin{align*}
    X_2/{\phi_4} &= \{ \{\L_1,\R_1\}, \{\L_2,\R_2\}, \{\Ltop\} \},
  \end{align*}
  and $\psi\in\Phi'_2$ is the equivalence relation such that
  $\L_1\not\equiv_{\psi}\L_2$.
  Figure~\ref{fig:RA} shows a transition sequence of $\calA$
  starting from an ID $(p_1,[d_3,d_0],d_3d_2d_0)$.
  As shown in the figure,
  $(p_1,[d_3,d_0],d_3d_2d_0)\vdash_{\!\calA}^*
  (q_2,[d_4,d_5],\varepsilon)$.
  Since $[d_4,d_5]\models\psi$,
  $(p_1,[d_3,d_0],\allowbreak d_3d_2d_0)\in L(\calA)$.
\end{example}
\begin{figure}
  \centering
  \begin{picture}(120,80)
    \put(0,75){\makebox(0,0){\footnotesize state}}
    \put(30,75){\makebox(0,0){\footnotesize registers}}
    \put(55,75){\makebox(30,0){\footnotesize stack}}

    \putconfig(0,60){p_1}{d_3,d_0}{d_3\,d_2\,d_0}
    \putconfig(0,40){q_1}{d_4,d_0}{d_2\,d_0}
    \putconfig(0,20){q_1}{d_4,d_0}{d_0}
    \putconfig(0,0){q_2}{d_4,d_5}{\varepsilon}

    \putruleA(0,40){\phi_1}
    \putruleA(0,20){\phi_4}
    \putruleA(0,0){\phi_3}
  \end{picture}
  \caption{A transition sequence of $\calA$ starting from an ID
  $(p_1,[d_3,d_0],d_3d_2d_0)$.}
  \label{fig:RA}
\end{figure}

\section{LTL model checking problem and valuations}
\label{sec:modelcheck}

We fix a finite set $\At$ of atomic propositions,
and let $\Sigma=2^{\At}$.

A \emph{valuation} for a $k$-RPDS $\calP=(P,\Delta)$ is
a function
$\Lambda:\ID_{\calP}\to\Sigma$,
which labels each ID of $\calP$ with
a subset of atomic propositions.
We extend the domain of $\Lambda$ to
$(\ID_{\calP})^{\infty}$ in the usual way;
i.e.,
$\Lambda(c_0c_1\ldots)=\Lambda(c_0)\Lambda(c_1)\ldots$\,.

\begin{definition}
The \emph{LTL model checking problem} for RPDS is defined
as:
\begin{description}
\item[Instance:]
  a $k$-RPDS $\calP=(P,\Delta)$,
  an LTL formula $\varphi$ over $\At$,
  a valuation $\Lambda:\ID_{\calP}\to\Sigma$,
  and an ID $c_0\in\ID_{\calP}$.
\item[Question:]
  Does every run
  $\rho\in(\ID_{\calP})^{\omega}$ of $\calP$
  with $\rho(0)=c_0$ satisfy
  $\Lambda(\rho)\models\varphi$\,?
\end{description}
\end{definition}

In the rest of the paper,
we fix a $k$-RPDS $\calP=(P,\Delta)$,
a valuation $\Lambda:\ID_{\calP}\to\Sigma$,
and
a starting ID~$c_0$.

\begin{definition}
We call $\Lambda:\ID_{\calP}\to\Sigma$ a \emph{regular valuation}
if the set $\{c\in\ID_{\calP}\mid A\in\Lambda(c)\}$ is
regular for every $A\in\At$.
\end{definition}

We assume that $\Lambda$ is a regular valuation
and
a $k$-RA $\calA_A=(Q_A,P,\xi_A,\delta_A)$ for each $A\in\At$
satisfying
$L(\calA_A)=\{c\in\ID_{\calP}\mid A\in\Lambda(c)\}$
is given.

\medskip

We also define the model checking problem and regular valuations
for PDS in the same way.
It is known that
the LTL model checking problem with regular valuations for PDS
is decidable~\cite[Theorem 3]{EKS03}.
The main objective of this paper is to show
a reduction of that problem
for RPDS
to the one for PDS\@.


However,
that problem for RPDS is undecidable in general,
as shown below.
\begin{theorem}\label{th:undecidable}
The LTL model checking problem with regular valuations for RPDS
is undecidable.
\end{theorem}
\begin{proof}
The universality problem for RA stated as follows
is known to be undecidable~\cite[Theorem 5.1]{NSV04}:
\emph{Instance}: a $k$-RA $\calA$, an initial state $q_0$,
and an initial assignment $\theta_0$.
\emph{Question}: Does $(q_0,\theta_0,w)\in L(\calA)$
for every $w\in D^*$?
We can reduce this problem to the model checking problem with
regular valuations for RPDS\@;
for given $\calA$, $q_0$, and $\theta_0$,
we can construct a $(k+1)$-RPDS $\calP$ and
a $(k+1)$-RA $\calA'$ that satisfy the following:
Let $\tdollar\in D$ be an arbitrary data value,
which is used as the stack bottom.
$\calP$ has $q_0$ as its only state.
$\calP$ does not alter
the first $k$ registers, and in every state transition,
it loads an arbitrary data value 
to the $(k+1)$-th register and pushes it into the stack.
Therefore, from any starting ID $(q_0,\theta,\tdollar)$,
$\calP$ can reach
an ID $(q_0,\theta',w\tdollar)$ for every $w\in D^*$,
where $\theta'(i)=\theta(i)$ for $i\in[k]$.
$\calA'$ is a modified version of $\calA$ that
does not use the $(k+1)$-th register and
satisfies
for any $\theta'_0$ with
$\theta'_0(i)=\theta_0(i)$ for $i\in[k]$,
$(q_0,\theta_0,w)\in L(\calA)$ iff
$(q_0,\theta'_0,w\tdollar)\in L(\calA')$.
Let $\At=\{A\}$ and $\Lambda$ be the regular valuation such that
$\Lambda(c)=\{A\}$ if $c\in L(\calA')$ and
$\Lambda(c)=\emptyset$ otherwise.
Let $c_0=(q_0,\theta'_0,\tdollar)$ for some $\theta'_0$ with
$\theta'_0(i)=\theta_0(i)$ for $i\in[k]$.
Let $\varphi=\Box A$.
Then, the answer to the model checking problem on
$\calP$, $\Lambda$, $\varphi$, and $c_0$ coincides with
the answer to the universality problem
on $\calA$, $q_0$, and $\theta_0$.
\qed
\end{proof}

\section{Freshness property}
\label{sec:freshness}

The method proposed in this paper requires that
the transition relation $\done_{\calP}$ of every $k$-RPDS
(including $k$-RA)
$\calP$
should have
the \emph{freshness property} stated as follows:
If $(p,\theta,du)\done_{\calP} (q,\theta',u')$
by a rule $r$
and the updated assignment $\theta'$ contains
a data value $d'$ not in $\theta\cup\{d\}$,
then $d'$ must be ``fresh''; i.e.,
$d'$ should have never appeared in the computation
from a starting configuration to $(p,\theta,du)$.
Since the set $D$ of data values is infinite,
such a fresh data value $d'$ always exists
whenever the rule $r$ can be applied to $(p,\theta,du)$ above.

To define the freshness property formally,
we slightly modify the semantics of $k$-RPDS so that
each stack cell keeps
the assignment at the time when the cell is pushed into the stack.
The assignments ``saved'' in the stack are used only for
choosing a fresh data value and do not affect
the behavior of an RPDS in other ways.
We redefine the set $\ID_{\calP}$ of IDs of
a $k$-RPDS $\calP=(P,\Delta)$ as
\[\ID_{\calP}=P\times\Theta_k\times(D\times\Theta_k)^*\]
and the transition relation $\done_{\calP}$ as follows:
\begin{equation*}
\begin{array}{l}
  (p,\phi)\to(q,\com)\in\Delta \quad
  (\theta,d,\theta')\models\phi \\
  \frsp(\theta'; d,\theta; \theta''\snd(u))
  \\ \hline
  (p,\theta,(d,\theta'')u)\done_{\calP} (q,\theta',
  \upds((d,\theta'')u,\theta',\com))
\end{array}
\end{equation*}
where
$\upds((d,\theta'')u,\theta',\com) = u$, $(d,\theta'')u$,
or $(\theta'(j),\theta')(d,\theta'')u$
if $\com=\pop$, $\Skip$, or $\push(j)$, respectively,
and
$\frsp(\theta'; d,\theta; \theta_n\ldots\theta_1)$
is a predicate
that is true iff
for each $i\in[k]$,
$\theta'(i)\in\theta\cup\{d\}$ or
$\theta'(i)\notin\theta_1\cup\ldots\cup\theta_n$.
That is,
each value $\theta'(i)$ in the updated assignment
is either
the value $\theta(l)$ of some register,
the value $d$ at the stack top,
or a fresh value that has not appeared in $\theta_n\ldots\theta_1$.

We say an ID
$(p,\theta_n,(d_{n-1},
\theta_{n-1})\allowbreak\ldots(d_1,\theta_1))\in\ID_{\calP}$
is \emph{proper} if
for $\forall i,j,l\in [n]$ with $i < j \le l$,
$d_i,\theta_i,\theta_j,\theta_l$ satisfy
$d_i\in\theta_i$ and
$\bigl(
\forall m\in[k]:
\theta_i(m)\notin\theta_j$ implies $\theta_i(m)\notin\theta_l
\bigr)$ and $\bigl(
d_i\notin\theta_j$ implies $d_i\notin\theta_l
\bigr)$.
%
%
Under the assumption of the freshness property,
every ID reachable from a proper starting ID is also proper.
%
We assume the starting ID $c_0$ 
given to the model checking problem
is proper. 

Note that
RA with the freshness property is similar to
session automata (SA)~\cite{BHLM14}, which is a special case
of fresh-register automata (FRA)~\cite{Tze11}.
An SA has the same structure as RA
but requires an input data value
to be either a value stored in some register or
a value not used before.
On the other hand,
RA in this paper can update
registers with data values not given as input
but chosen arbitrary from values satisfying a guard condition,
and the freshness constraint is imposed only on values not given as input.

\section{PDS simulating RPDS}
\label{sec:rpds-pds}

\subsection{Bisimulation relation between an RPDS and a PDS}

The bisimulation equivalence~\cite{CGKPV18-Bisim} is a basic notion
to capture the equivalence of behaviors of
two state transition systems.
We review the definition of
the notion in this subsection.
We will show the construction of
a PDS bisimulation equivalent to a given RPDS
in Section~\ref{sec:RPDStoPDS}.

\begin{definition}
For an RPDS $\calP=(P,\Delta)$ and a PDS $\calP'=(P',\Delta')$,
we call a relation $R\subseteq \ID_{\calP}\times\ID_{\calP'}$
a \emph{bisimulation relation} between $\calP$ and $\calP'$
if $R$ satisfies the following:
\begin{enumerate}[label=(\arabic*)]
\item For every $c_1,c_2\in\ID_{\calP}$ and $c'_1\in\ID_{\calP'}$,
  if $c_1\done_{\calP} c_2$ and $(c_1,c'_1)\in R$,
  then $\exists c'_2\in\ID_{\calP'}: c'_1\done_{\calP'} c'_2$
  and $(c_2,c'_2)\in R$.
\item For every $c_1\in\ID_{\calP}$ and $c'_1,c'_2\in\ID_{\calP'}$,
  if $c'_1\done_{\calP'} c'_2$ and $(c_1,c'_1)\in R$,
  then $\exists c_2\in\ID_{\calP}: c_1\done_{\calP} c_2$
  and $(c_2,c'_2)\in R$.
\end{enumerate}
For an RA $\calA=(Q,I,\xi,\delta)$ and an NFA $\calA'=(Q',I',F,\delta')$,
we call a relation $R\subseteq \IDA\times\ID_{\!\calA'}$
a bisimulation relation between $\calA$ and $\calA'$
if $R$ is a bisimulation relation between
$(Q,\delta)$ and $(Q',\delta')$ and
also satisfies:
\begin{enumerate}[label=(\arabic*),start=3]
\item
  If $(c,c')\in R$,
  then $c\in\Acc_{\calA}$ iff $c'=(q,\varepsilon)$ and $q\in F$.
\end{enumerate}
An RPDS $\calP$ and a PDS $\calP'$
(or an RA $\calA$ and an NFA $\calA'$)
are \emph{bisimulation equivalent}
if there is a bisimulation relation between them.
\end{definition}

By definition, we obtain the following propositions.

\begin{proposition}
If there is a bisimulation relation $R$ between
an RA $\calA$ and an NFA $\calA'$,
then
for any pair $(c,c')\in R$,
$c\in L(\calA)$ iff $c'\in L(\calA')$.
\end{proposition}

\begin{proposition}
Let $\calP$ be an RPDS and $\calP'$ be a PDS\@.
Let $\Lambda:\ID_{\calP}\to\Sigma$ and
$\Lambda':\ID_{\calP'}\to\Sigma$ be their valuations
and $c_0\in\ID_{\calP}$ and $c'_0\in\ID_{\calP'}$ be
their IDs.
If there is a bisimulation relation $R$
between $\calP$ and $\calP'$ such that
$(c_0,c'_0)\in R$ and
$\Lambda(c)=\Lambda'(c')$
for every
pair $(c,c')\in R$, then:
\begin{enumerate}[label=(\roman*)]
\item For every run $\rho$ 
of $\calP$
with $\rho(0)=c_0$,
there exists a run $\rho'$ 
of $\calP'$
with $\rho'(0)=c'_0$ such that
$\Lambda(\rho)=\Lambda'(\rho')$.
\item
For every run $\rho'$ of $\calP'$
with $\rho'(0)=c'_0$,
there exists a run $\rho$ of $\calP$
with $\rho(0)=c_0$ such that
$\Lambda(\rho)=\Lambda'(\rho')$.
\end{enumerate}
\end{proposition}

%
%

\subsection{Composition of equivalence relations}
\label{sec:composition}

Let $\COMPBL$ and $\COMPBLT$ be the binary predicates over $\Phi_k$ defined as:
\begin{align*}
  \phi_1\COMPBL\phi_2 \ &:\Leftrightarrow \
      \bigl(\R_i \equiv_{\phi_1} \R_j \text{ iff } \L_i \equiv_{\phi_2} \L_j
      \ \text{for } i,j\in[k]\bigr).
  \\
  \phi_1\COMPBLT\phi_2 \ &:\Leftrightarrow \
      \bigl(\phi_1\COMPBL\phi_2 \ \text{and} \\
      &\quad\;
      (\R_i \equiv_{\phi_1} \Ltop \text{ iff } \L_i \equiv_{\phi_2} \Ltop
      \ \text{for } i\in[k])\bigr).
\end{align*}
Intuitively,
$\phi_1\COMPBL\phi_2$ represents the \emph{composability} of
$\phi_1$ and $\phi_2$;
for $\phi_1$ and $\phi_2$ satisfying
$(\theta_1,d_1,\theta_2)\models\phi_1$ and
$(\theta_2,d_2,\theta_3)\models\phi_2$
for some $\theta_1,d_1,\theta_2,d_2,\theta_3$,
we will define (after Example~\ref{example:composable})
the composition
$\phi_1\COMP\phi_2$ that satisfies
$(\theta_1,d_1,\theta_3)\models\phi_1\COMP\phi_2$
(under the assumption on the freshness property),
and
$\phi_1\COMPBL\phi_2$ represents
the condition
``$(\theta_1,d_1,\theta_2)\models\phi_1$ and
$(\theta_2,d_2,\theta_3)\models\phi_2$
for some $\theta_1,d_1,\theta_2,d_2,\theta_3$.''
Similarly,
$\phi_1\COMPBLT\phi_2$ represents the condition
``$(\theta_1,d,\theta_2)\models\phi_1$ and
$(\theta_2,d,\theta_3)\models\phi_2$
for some $\theta_1,d,\theta_2,\theta_3$.''
%

\begin{example}\label{example:composable}
  The equivalence relations $\phi_0$ and $\phi_1$ shown in
  Example~\ref{example:RPDS} satisfy
  $\phi_0\COMPBL\phi_1$
  because $\R_1\not\equiv_{\phi_0}\R_2$ and
  $\L_1\not\equiv_{\phi_1}\L_2$.
  However,
  $\phi_0\COMPBLT\phi_1$ does not hold
  because $\R_2\equiv_{\phi_0}\Ltop$ but
  $\L_2\not\equiv_{\phi_1}\Ltop$.
  For $\phi_2$ and $\phi_3$
  also shown in Example~\ref{example:RPDS},
  both
  $\phi_0\COMPBLT\phi_3$ and
  $\phi_1\COMPBLT\phi_2$ hold.
\end{example}

For $\phi_1,\phi_2\in\Phi_k$ with $\phi_1\COMPBL\phi_2$,
the \emph{composition} $\phi_1\COMP\phi_2$ of them
is the equivalence relation in $\Phi_k$ that satisfies the following:
\begin{alignat}{2}
  \L_i \equiv_{\phi_1\COMP\phi_2} \L_j :\Leftrightarrow{} &
  \L_i \equiv_{\phi_1} \L_j \nonumber\\
  &\text{for } i,j\in [k]\cup\{\Top\},\\
  \R_i \equiv_{\phi_1\COMP\phi_2} \R_j :\Leftrightarrow{} &
  \R_i \equiv_{\phi_2} \R_j
  \quad\text{for } i,j\in [k],\\
  \L_i \equiv_{\phi_1\COMP\phi_2} \R_j :\Leftrightarrow{} & 
  (\exists l\in [k] :{} \L_i \equiv_{\phi_1} \R_l
  \mathrel{\land}
  \L_l \equiv_{\phi_2} \R_j) \nonumber\\
  &\text{for } i\in[k]\cup\{\Top\},\ j\in [k].
  \label{eq:comp3}
\end{alignat}
%
By definition,
if an ID $(p,\theta_3,\allowbreak
(d_2,\theta_2)\allowbreak(d_1,\theta_1)u)$
is proper, 
$(\theta_1,d_1,\theta_2)\models\phi_1$ and
$(\theta_2,d_2,\theta_3)\models\phi_2$,
then
$(\theta_1\,d_1,\theta_3)\models\phi_1\COMP\phi_2$.
Guaranteeing this property is
the main purpose of the freshness property
and the properness of IDs:
If the properness of the above ID is not assumed,
and
$d_1\notin\theta_2$ and
$\theta_3(j)\notin\theta_2\cup\{d_2\}$ for some $j\in[k]$,
then
either
$d_1=\theta_3(j)$ or $d_1\ne\theta_3(j)$ holds
(and only the latter satisfies
$(\theta_1\,d_1,\theta_3)\models\phi_1\COMP\phi_2$).
This uncertainty prevents
a PDS from simulating an RPDS\@:
When a PDS $\calP'$ simulating an RPDS is popping off the stack top
in an ID corresponding to the above ID,
without assuming the properness of the ID,
$\calP'$ cannot know whether or not
the data value $d_1$ in the new stack top of the RPDS
is contained in the current register assignment~$\theta_3$.

Similarly to the composition $\circ$,
for $\phi_1,\phi_2\in\Phi_k$ with $\phi_1\COMPBLT\phi_2$,
we define $\phi_1\COMPT\phi_2$
as the same as $\phi_1\COMP\phi_2$ except that
the Equation~(\ref{eq:comp3}) is replaced with the following
(\ref{eq:compT}):
\begin{alignat}{2}
  \L_i \equiv_{\phi_1\COMPT\phi_2} \R_j
  :\Leftrightarrow{} & 
  (\exists l\in [k] :{} \L_i \equiv_{\phi_1} \R_l \mathrel{\land}
  \L_l \equiv_{\phi_2} \R_j) \nonumber\\
  &\quad\:\lor{}
  (\L_i\equiv_{\phi_1}\Ltop \mathrel{\land} \Ltop\equiv_{\phi_2}\R_j)
  \nonumber\\
  &\text{for } i\in[k]\cup\{\Top\},\ j\in [k].
  \label{eq:compT}
\end{alignat}
By definition, $\COMP$ and $\COMPT$ are associative.

Let $\EQj{\phi}{j}$ for $\phi\in\Phi_k$ and $j\in[k]$ be
the equivalence relation defined as follows:
$\forall i,l\in[k]:
\bigl(\L_i\equiv_{\EQj{\phi}{j}}\L_l$ iff $\R_i\equiv_{\phi}\R_l\bigr)
\land
\bigl(\L_i\equiv_{\EQj{\phi}{j}}\Ltop$ iff $\R_i\equiv_{\phi}\R_j\bigr)
\allowbreak
\land
(\L_i\equiv_{\EQj{\phi}{j}}\R_i)$.
%
The intention of the above definition is to guarantee that
$(\theta',\theta'(j),\theta')\models\EQj{\phi}{j}$
whenever $(\theta,d,\theta')\models\phi$.

\subsection{Construction of PDS simulating RPDS}
\label{sec:RPDStoPDS}

For a $k$-RPDS $\calP=(P,\Delta)$,
we construct a PDS $\calP'=(P',\Delta')$
bisimulation equivalent to $\calP$.
The set of states of $\calP'$ is
$P'=P\times\Phi_k$,
and the stack alphabet of $\calP'$ is $\Phi_k$.
%
$\calP'$ must simulate $\calP$ without keeping data values in the stack.
When popping off the stack top,
$\calP'$ must know whether or not the data value
in the new stack top of $\calP$ equals the \emph{current}
value of each register.
For this purpose, $\calP'$ keeps an abstract ``history''
of the register assignments
represented by a sequence of equivalence relations
in the stack,
which tells whether each of the data values
in the stack of $\calP$ equals the current value of each register.
The second component of each state of $\calP'$
is ``the last element'' of the abstract history,
which represents the accumulated updates
since the current stack top has been pushed into the stack.
Because a PDS cannot replace the symbol at the new stack top
when $\pop$ or $\Skip$ is performed,
$\calP'$ keeps the last element of the history in its finite state
and updates it in every transition.
%

For example, configuration
$(q,\theta_2,(d_1,\theta_1)(d_0,\theta_0))$
of $\calP$ is simulated by configuration
$((q,\phi_2),\phi_1\phi_0)$ of $\calP'$,
where $(\theta_0,d_0,\theta_1)\models\phi_1$ and
$(\theta_1,d_1,\theta_2)\models\phi_2$.
Equivalence relation $\phi_1$ abstractly represents
the updates of assignments between
when $d_0$ has been pushed and when $d_1$ has been pushed.
Similarly,
$\phi_2$ represents
the updates of assignments since
$d_1$ has been pushed.

The set $\Delta'$ of transition rules of $\calP'$ is
the smallest set satisfying
the following inference rules.
\begin{align*}
&
\begin{array}{l}
  (q,\phi_3)\to(q',\Skip)\in\Delta \quad
  \phi_1\COMPBL\phi_2 \quad
  \phi_2\COMPBLT\phi_3
  \\ \hline
  ((q,\phi_2),\phi_1) \to ((q',\phi_2\COMPT\phi_3),\Skip) \in\Delta'
\end{array}
\\[\medskipamount]
&
\begin{array}{l}
  (q,\phi_3)\to(q',\pop)\in\Delta \quad
  \phi_1\COMPBL\phi_2 \quad
  \phi_2\COMPBLT\phi_3
  \\ \hline
  ((q,\phi_2),\phi_1)\to((q',\phi_1\COMP(\phi_2\COMPT\phi_3)),\pop)
  \in\Delta'
\end{array}
\\[\medskipamount]
&
\begin{array}{l}
  (q,\phi_3)\to(q',\push(j))\in\Delta \quad
  \phi_1\COMPBL\phi_2 \quad
  \phi_2\COMPBLT\phi_3
  \\ \hline
  ((q,\phi_2),\phi_1)\to((q',\EQj{\phi_3}{j}),\push(\phi_2\COMPT\phi_3))
  \in\Delta'
\end{array}
\end{align*}
In the above inference rules, $\phi_2\COMPT\phi_3$
represents the accumulation of the update of registers
by $\phi_3$ into $\phi_2$.
When $\pop$ is performed,
$\phi_1$ at the stack top is composed into
$\phi_2\COMPT\phi_3$,
which then represents the accumulated updates
since the new stack top has been pushed into the stack.
When $\push$ is performed,
$\phi_2\COMPT\phi_3$ is pushed into the stack because
the current assignment is ``the assignment when the stack top was pushed
into the stack.''
In this case, $\calP'$ sets the second component of the state
to $\EQj{\phi_3}{j}$,
which represents the current assignment
(which is the result of the update by $\phi_3$)
equals the assignment saved in the stack top.

In the same way,
we construct an NFA $\calA'_A$, as a PDS with pop rules only,
from $\calA_A=(Q_A,P,\allowbreak\xi_A,\delta_A)$ for each $A\in\At$.
Moreover, let the set of initial states and the set of final states
of $\calA'_A$ be
$P\times\Phi_k$ and
$\{ (p,\phi) \mid (p,\lat(\phi))\in\xi_A \}$,
respectively, for each $A\in\At$.
Let $\ID_{\!\calA}=\bigcup_{A\in\At} \ID_{\!\calA_A}$
and $\ID_{\!\calA'}=\bigcup_{A\in\At} \ID_{\!\calA'_A}$.

\medskip

We assume that the stack $u_0$ in
the given starting ID
$c_0=(p_0,\theta'_0,u_0)$ is not empty,
and let $(d_0,\theta_0)$ be the last (bottom-most) element of $u_0$.
Note that the last element of the stack of
\emph{every} ID reachable from $c_0$ equals $(d_0,\theta_0)$,
because an ID with empty stack has no successor
and thus any computation cannot alter the stack bottom.

Let $R\subseteq \ID_{\!\calA}\times\ID_{\!\calA'}$
be the smallest relation satisfying the following:
$((p,\theta_n,u),\allowbreak((p,\phi_n),v))\in R$
for $u=(d_{n-1},\theta_{n-1})\ldots(d_1,\theta_1)$ and
$v=\phi_{n-1}\ldots\phi_1$
iff
  $(p,\theta_n,u)$ is proper, 
  $(d_1,\theta_1)=(d_0,\theta_0)$ (or $u=\varepsilon$), and
  $\forall i\in[n]: (\theta_{i-1},d_{i-1},\theta_i)\models\phi_i$.
(Remember that $(d_0,\theta_0)$ is the stack bottom in
the starting ID~$c_0$.)

By definition,
$R$ is functional; that is,
each $(p,\theta_n,u)\in\IDA$ has
exactly one $((p,\phi_n),v)\in\ID_{\!\calA'}$
that satisfies
$((p,\theta_n,u),((p,\phi_n),v))\in R$.
Let $R(c)$ for $c\in\IDA$
denote the unique ID $c'\in\ID_{\!\calA'}$
that satisfies $(c,c')\in R$.

Let
$R_{\calP}=R\cap(\ID_{\calP}\times\ID_{\calP'})$
and
$R_{\calA_A}=R\cap(\ID_{\!\calA_A}\times\ID_{\!\calA'_A})$
for each $A\in\At$.
%
We have proved the following proposition using
the Coq proof assistant~\cite{BC04}.\footnote{%
The proof scripts are available at
\url{https://github.com/ytakata69/rpds-to-pds-proof}.}
\begin{proposition}\label{prop:bisimP}
$R_{\calP}$ is a bisimulation relation between
$\calP$ and $\calP'$.
$R_{\calA_A}$ is a bisimulation relation between
$\calA_A$ and $\calA'_A$
for each $A\in\At$.
\end{proposition}

The reduction is completed by
letting $c'_0=R(c_0)$.

\begin{example}
  Let $\calP'$ be the PDS obtained from the 2-RPDS $\calP$ shown in
  Example~\ref{example:RPDS}.
  Let $\phi_0,\ldots,\phi_3\in\Phi_2$ be
  the equivalence relations also shown in Example~\ref{example:RPDS}.
  By the above inference rules,
  $\calP'$ has rules
  $((p_1,\phi''),\phi')\to
  ((p_1,\EQj{\phi_1}{1}),\push(\phi''\COMPT\phi_1))$
  (obtained from $r_2$)
  and
  $((p_1,\phi''),\phi')\to
  ((p_1,\phi'\COMP(\phi''\COMPT\phi_1)),\pop)$
  (obtained from $r_3$)
  for each $\phi'$ and $\phi''$ such that
  $\phi'\COMPBL\phi''$ and $\phi''\COMPBLT\phi_1$.
  Let $\phi_5,\phi_6\in\Phi_2$ be the equivalence relations
  defined by the following quotient sets:
  \begin{align*}
    X_2/{\phi_5} &= \{ \{\L_1,\R_1,\Ltop\}, \{\L_2,\R_2\} \}, \\
    X_2/{\phi_6} &= \{ \{\L_1,\R_1\}, \{\L_2,\R_2,\Ltop\} \}.
  \end{align*}
  Because $\phi_0\COMPBL\phi_5$,
  $\phi_5\COMPBLT\phi_1$,
  $\phi_5\COMPT\phi_1=\phi_1$, and
  $\EQj{\phi_1}{1}=\phi_5$,
  $\calP'$ has rule
  $r'_2=((p_1,\phi_5),\phi_0)\to
  ((p_1,\phi_5),\push(\phi_1))$.
  Similarly, because $\phi_1\COMPBL\phi_5$ and
  $\phi_1\COMP\phi_1=\phi_1$,
  $\calP'$ has rule
  $r'_3=((p_1,\phi_5),\phi_1)\to
  ((p_1,\phi_1),\pop)$.
  By these two rules, $\calP'$ has
  the following transition sequence:
  \begin{align}
    ((p_1,\phi_5),\phi_0\phi_6)
    &\done_{\calP'}
    ((p_1,\phi_5),\phi_1\phi_0\phi_6) \nonumber \\
    &\done_{\calP'}
    ((p_1,\phi_1),\phi_0\phi_6).
    \label{eq:PDStrans}
  \end{align}
  The following is a part of the transition sequence
  of $\calP$ in Figure~\ref{fig:RPDS}, in which
  each stack cell is augmented by the assignment at the time
  when the cell has been pushed. (See Section~\ref{sec:freshness}.)
  \begin{align}
    &(p_1,[d_2,d_0],(d_2,[d_2,d_0])(d_0,[d_1,d_0])) \nonumber \\
    &\done_{\calP}
    (p_1,[d_3,d_0],(d_3,[d_3,d_0])(d_2,[d_2,d_0])(d_0,[d_1,d_0]))
    \nonumber \\
    &\done_{\calP}
    (p_1,[d_4,d_0],(d_2,[d_2,d_0])(d_0,[d_1,d_0])).
    \label{eq:RPDStrans}
  \end{align}
  Let $c_1,c_2,c_3$ be the IDs in the sequence~(\ref{eq:RPDStrans}),
  respectively.
  We can see that the IDs 
  in the sequence~(\ref{eq:PDStrans}) are
  $R(c_1)$, $R(c_2)$, and $R(c_3)$, respectively.
  (See the paragraphs before Proposition~\ref{prop:bisimP}
   for the definition of $R$.)
\end{example}

\subsection{Time complexity}

Consider the LTL model checking problem on 
$\calP=(P,\Delta)$, 
$\varphi$, 
$\Lambda$, and 
$c_0$,
where $\Lambda$ is a regular valuation and
is represented by $k$-RA $\calA_A=(Q_A,P,\xi_A,\delta_A)$
for $A\in\At$.
Let $\calP'=(P',\Delta')$ and $\calA'_A=(Q'_A,P',F_A,\delta'_A)$
be the PDS and NFA obtained from $\calP$ and $\calA_A$
in the last subsection.
Applying the LTL model checking method for PDS
in~\cite{EKS03} to $\calP'$,
we can solve the model checking problem on
$\calP, \varphi, \Lambda, c_0$
in $O(|P'|^2\cdot|\Delta'|\cdot
\prod_{A\in\At}|Q'_A|
\cdot 2^{O(|\varphi|)})$
time,
if $\calA'_A$ for each $A\in\At$ is backward-deterministic.\footnote{%
  In \cite{EKS03}, each FA for a regular valuation
  is defined as the one that reads the \emph{reverse} of the stack contents,
  and thus a deterministic FA in \cite{EKS03} is
  a backward-deterministic FA in our settings.}
By the construction,
$|P'|=|P|\cdot|\Phi_k|$ and
$|Q'_A|=|Q_A|\cdot|\Phi_k|$.
Moreover,
$|\Delta'|\le |\Delta|\cdot|\Phi_k|^2$,
because
in the inference rules defining $\Delta'$,
we choose two equivalence relations $\phi_1$ and $\phi_2$
for each transition rule in $\Delta$.
$|\Phi_k|$ equals the $(2k+1)$-th Bell number and thus
$|\Phi_k|=2^{O(k\log k)}$.
We can assume that $|\At|\le |\varphi|$ and thus
$\prod_{A\in\At}|Q'_A| \le \prod_{A\in\At}|Q_A|\cdot|\Phi_k|^{|\varphi|}$.
Therefore the time complexity of the proposed method is
exponential in $k$ and $|\varphi|$ and polynomial in
$|P|$ and $|\Delta|$ and $\prod_{A\in\At}|Q_A|$.
Note that if the NFA $\calA'_A$ for some $A\in\At$
is not backward-deterministic,
we have to apply backward-determinization to $\calA'_A$,
which may increase the time complexity of the proposed method.
Also note that
$\calA'_A$ is not necessarily backward-deterministic
even when $\calA_A$ is backward-deterministic in the sense
defined in~\cite{STS21-IEICE-dec}.

\section{Conclusion}
\label{sec:conclusion}

This paper proposed a method for
solving the LTL model checking problem for RPDS
with regular valuations,
in which the problem is reduced to the same problem for PDS\@.
In contrast to the method for the same problem
proposed in~\cite{STS21-IEICE-dec},
the method in this paper does not require
RA used for a regular valuation to be
deterministic or backward-deterministic.
On the other hand, the method in this paper requires
every RPDS and RA has the freshness property instead.
This paper also showed that
the LTL model checking problem for RPDS with regular valuations
defined by general RA is undecidable,
and thus the freshness constraint is essential in this method.

\bibliography{rpds}

\end{document}